\documentclass[12pt]{article}
\usepackage[active]{srcltx}
\usepackage{amsmath,amsthm,amssymb,bbm,hyperref,fullpage}
\usepackage{mathabx}%for widecheck
\usepackage[utf8]{inputenc}
\usepackage{stmaryrd}
\usepackage{mathpartir} % For proof trees
\usepackage{tikz-cd}
\usepackage{url}
\newtheorem{defn}{Definition}
\newtheorem{lemma}{Lemma}
\newtheorem{prop}{Proposition}
\newtheorem{thm}{Theorem}
\newcommand{\F}{\ensuremath{\mathbbm{F}}}
\newcommand{\I}{\ensuremath{\mathbbm{I}}}

\newcommand{\psh}[1]{\widehat{#1}}
\newcommand{\fct}[1]{\widecheck{#1}}
\newcommand{\two}{\ensuremath{\mathbbm{2}}}
\newcommand{\op}{\ensuremath{{}^{\operatorname{op}}}}
\newcommand{\um}{\operatorname{-}}
\newcommand{\List}{\operatorname{list}}
\newcommand{\nil}{\operatorname{nil}}
\newcommand{\comp}{\operatorname{comp}}
\newcommand{\Path}{\operatorname{Path}}

\newcommand{\id}{\operatorname{Id}}
\newenvironment{acknowledgement}{\paragraph{Acknowledgement}}{}
\begin{document}
\author{Bas Spitters\\ Aarhus University}
\title{Cubical sets and the topological topos}
\maketitle
\abstract{Coquand's cubical set model for homotopy type theory
    provides the basis for a computational interpretation of the
    univalence axiom and some higher inductive types, as implemented
    in the \emph{cubical} proof assistant. This paper contributes to
    the understanding of this model. We make three contributions:
\begin{enumerate}
\item Johnstone's topological topos was created to present the
  geometric realization of simplicial sets as a geometric
  morphism between toposes. Johnstone shows that simplicial sets
  classify strict linear orders with disjoint endpoints and that
  (classically) the unit interval is such an order.
 Here we show that it can also be a target for cubical
  realization by showing that Coquand's cubical sets classify the
  geometric theory of flat
  distributive lattices. As a side result, we obtain a simplicial realization of a
  cubical set.
\item Using the internal `interval' in the topos of cubical sets, we
  construct a Moore path model of identity types.
\item We construct a premodel structure internally in the
  cubical type theory and hence on the fibrant objects in
  cubical sets.
\end{enumerate}
}

\section{Introduction}
Simplicial sets from a standard framework for homotopy theory. The
topos of simplicial sets is the classifying topos of the theory of
strict linear orders with endpoints.
Cubical sets turn out to be more amenable to a
constructive treatment of homotopy type theory. 
We consider the cubical set model in~\cite{Cubical}.
This consists of symmetric cubical sets with connections
($\wedge,\vee$), reversions ($\neg$) and diagonals. 
In fact, to present the geometric realization clearly, we will leave
out the reversions. We now give a precise definition.

\section{Classifying topos and geometric realization}
\paragraph{Cube category}
 Consider the monad $DL$ on  the category of finite sets with all
maps which assigns to each
finite set $F$ the \emph{finite} set of the free distributive lattice on $F$. That this
set is finite can be seen using the disjunctive normal form.
The \emph{cube category} is the Kleisli category for the monad $DL$.
The diagonals above refer to the use of all maps, as opposed to only
the injective ones; cf~\cite{pitts2014equivalent}.

\paragraph{Lawvere theory}
For each algebraic (=finite product) theory $T$, the
Lawvere theory $\Theta\op_{T}$ is the opposite of the category of free finitely
generated models. This is the classifying category for $T$:
models of $T$ in any finite product category category $E$ correspond to product-preserving
functors from the syntactic category $C_{fp}[T]$ to $E$.
For a monad $T$ on finite sets, the Kleisli category $KL_{T}$ is precisely the
\emph{opposite} of the Lawvere theory:
maps $I\to T(J)$ are equivalent to $T$-maps $T(I)\to T(J)$ since each such map is completely
determined by its behavior on the atoms, as $T(I)$ is free.

\paragraph{Cube category, contravariantly}
Let $\two$ be the poset with two elements, one smaller than the other. Consider $\Box$, the
full subcategory of Cat consisting of the powers of
\two.
We have a Stone duality between finite posets and finite distibutive
lattices with $\two$ as dualizing object.
This duality is given by a functor $J=\hom_{\mathrm{fDL}}(-,\two)$ from finite distributive
lattices to the opposite of the category of finite posets~\cite{birkhoff1937rings}\cite{Wraith}. This functor sends a distributive lattice to its
join-irreducible elements. It's inverse is the functor $\hom_{\mathrm{poset}}(-,\two)$ which sends a poset to its the
distributive lattice of lower sets. This restricts to a duality
between free finitely generated distributive lattices and powers
of $\two$ (the copowers of DL1).

\paragraph{De Morgan algebras}
There are two more specific dualities~\cite{cornish1977coproducts,cornish1979coproducts}:\\
Finite involutive posets and finite DeMorgan-algebras, with dualizing object $\two^2$.\\
Finite involutive posets such that $x\leq \neg x$ or $\neg x\leq x$
and finite Kleene algebras, with dualizing object 3 (three points on a line).

The natural involution on $\two$ provides us with an involutive poset
and hence, dually, with a De Morgan algebra. 
Every free finitely generated DeMorgan algebra on $n$ generators is a free distributive lattice
on $2n$ generators. We obtain a duality with the category \emph{even} powers of $\two$
and maps preserving the De Morgan involution~\cite{cornish1977coproducts}.

Although it is natural to consider De Morgan algebras or Kleene
algebras in the implementation, we will focus on distributive lattices in what follows, mainly in view
of the geometrical realization; see Section~\ref{sec:demorg-algebr-invol}.

\paragraph*{Classifying topos}
Objects of the topos $\psh{\Box}\equiv\psh{\Theta_{DL}}$ are called
\emph{cubical sets}. The following theorem by Johnstone-Wraith~\cite[Thm
5.22]{johnstone1978algebraic} tells us what this topos classifies.
\begin{thm}\label{JW} Let $\Theta_T\op$ be a Lawvere
  theory. Then $\psh{\Theta_T}$ classifies flat $T$-models. 
\end{thm}
Here a \emph{flat} $T$-models is one that can be expressed as a
filtered colimit of free models. Flatness is a geometric notion~\cite[Thm 5.22]{johnstone1978algebraic}.

Below we will compute what this means for the algebraic theory of distributive lattices,
but first we indicate how this theorem can be proved.
To obtain the classifying topos for an algebraic theory $T$, we first need to
complete with finite limits, i.e.\ to consider the 
classifying category $C_{fl}$ as the
\emph{opposite} of finitely \emph{presented} $T$-algebras.
Then $C_{fl}^{op}\to Set$, i.e.\ functors on finitely presented
$T$-algebras, is the classifying topos for the theory of $T$-algebras. This topos contains a generic $T$-algebra.
$T$-algebras in any topos $\mathcal{F}$ correspond to \emph{left exact left adjoint}
functors from the classifying topos to $\mathcal{F}$.

% Caramello: category of finitely generated (equivalently, finitely
% presented) rings and ring homomorphisms.
% http://www.oliviacaramello.com/Unification/TechnicalExplanation.html

Let $FG$ be the category of \emph{free finitely generated} $T$-algebras
and let $FP$ the category of \emph{finitely presented} ones. We have a fully
faithful functor $f:FG \to FP$. This gives a % (n essential) 
geometric morphism $\phi:\fct{FG}\to\fct{FP}$. Since $f$ is fully faithful, $\phi$ is an embedding~\cite[A4.2.12(b)]{elephant}. %(MM p377/Elephant A4.2.12(b)).
%
% Note that the \emph{direct} image of an inclusion preserves exponents
% (note: for open maps this is the \emph{inverse} image ).\\
% ( The direct image can be computed as (A4.1.4) the right Kan extension: 
% \[\phi^*(A)=\lim (A\downarrow \phi)\to FG \to Set\]
% where $(A\downarrow \phi)$ consists of pairs $(B,f)$ such that $f:B\to
% \phi(A)$. )\\
%
The subtopos $\fct{FG}$ of the classifying topos for $T$-algebras is
given by a quotient theory, the theory of the model
$\phi^*M$. This model is given by pullback and thus is equivalent to the
canonical $T$-algebra $\I(m):=m$ for each $m\in FG$. 

\paragraph{Flat distributive lattices}
Let $D$ be a distributive lattice in a topos $\mathcal{X}$. Then by the standard
construction in Lawvere theories, define $S_D:\Theta\op_{DL}\to
\mathcal{X}$ on objects by $S_D(n):=D^n$ and for a map $\phi:n\to DL m$ define a
map $S_D\phi:D^m\to D^n$. This is well-defined since a distributive lattice is an algebra for the DL-monad. It
follows that $S_D:\Box\to \mathcal{X}$ is a cocubical object in $\mathcal{X}$.

%\url{http://ncatlab.org/nlab/show/flat+functor}
A Set-valued functor is $E : C \to Set$ is flat if it is filtering~\cite[VII.6]{maclanemoerdijk}:
\begin{description}
\item[inhabited] There is at least one object $c \in C$ such that $E(c)$ is an inhabited set.
\item[transitivity] For objects $c,d \in C$ and elements $y \in E(c)$,
  $z \in E(d)$, there exists an object $b \in C$, morphisms $\alpha :
  b \to c$, $\beta : b \to d$ and an element $w \in E(b)$ such that
  $E(\alpha)(w)=y$ and $E(\beta)(w)= z$.
\item[freeness] For two parallel morphisms $\alpha, \beta : c \to d$ and $y \in E(c)$ such that $E(\alpha)(y) = E(\beta)(y)$, there exists a morphism $\gamma : b \to c$ and an element $z \in E(b)$ such that $\alpha \circ \gamma = \beta \circ \gamma$ and $E(\gamma)(z) =y$.
\end{description}
Specializing the general definition of a flat model to $T$-algebras
($C=\Theta\op_{T}$), a $T$-algebra $D$, we observed that the first two
conditions always hold:
\begin{description}
\item[inhabited] $D^1$ is inhabited.
\item[transitivity] given $d\in D^n$ and $d'\in D^m$, then $d,d'\in D^{n+m}$ shows transitivity.
\end{description}

So, a $D$ is \emph{flat} if for all
$\alpha,\beta:n\to Tm$ and $d\in D^m$ st $\alpha d=\beta d$, there
exists $\gamma:m\to Tk$ such that $\alpha\gamma=\beta\gamma$ and
there exists a $d'\in D^k$ such that $\gamma d'=d$.

Put more simply, if we have two $n$-ary $T$-expressions (`polynomials') $\alpha,\beta$
which when applied to $d$ are equal, then there exists $d'$ such
that $\gamma d'=d$ and $\alpha,\beta$ are both constructed from
$\gamma$. 
% Equality of functions is pointwise, but we still need to consider
% general n, as all the $\gamma_n$ need to equalize simultaneously.

Flat functors generalizes the abstract definition of flat
modules~\cite[VII]{maclanemoerdijk}. The definition above is similar
to the elementary definition of flat modules, with the difference that we 
lack subtraction, and hence equalities between terms cannot be replaced
by being equal to 0.

Vickers~\cite{vickers2007locales} observes that being flat is a
geometric notion. The following is implicit in~\cite{johnstone1978algebraic}.
%In particular,
%being free is geometric with an existential quantification over a
%countable type of functions $n\to DL m$ varying over $n,m$.

\begin{lemma}
  A (free) finitely generated $T$-algebra is flat. Hence, the generic
$T$-algebra is a flat model.
\end{lemma}
\begin{proof}
  For readability, we fix $T$ to be the theory of distributive lattices.
  Every element $d$, e.g.\ $(x,x\vee y,y\wedge x)$, is the image under
  a map $\gamma$ of a list of generators $d'$, $(x,y)$ in the example. 
  We also have $\alpha\gamma=\beta\gamma$, as $\gamma$ is completely
  determined by its behavior on the generators. This shows that
  finitely generated distributive lattices are free.

  Since, being flat is a geometric statement it also holds for the
  generic element $\I$.
\end{proof}

\begin{lemma}
Flat distributive lattices have the disjunction property: If $a\vee
b=1$, then $a=1$ or $b=1$.
\end{lemma}
\begin{proof}
Consider $d=(a,b)$ and $\alpha=x\vee y$
and $\beta=1$. There are essentially two possibilities for $\gamma$:
$\gamma(x)=(x,1)$ or $\gamma(x)=(1,x)$. From $\gamma d'=d$, it follows
that $a=1$ or $b=1$.  
\end{proof}

This disjunction property crucial in CTT, as it is used to prove that
we have a lattice homomorphism $\I\to\F$~\cite[3.1]{connections}\cite{GCTT}.

By Diaconescu's theorem, flat functors correspond to geometric morphisms. In fact~\cite{vickers2007locales}, a presheaf topos is the classifying topos for flat functors. The generic element is the Yoneda flat functor.

\begin{prop}\label{prop:flatfree}
  Every flat functor $\Box\to \mathcal{X}$ is $S_D$ for some flat $D$ in $\mathcal{X}$.
\end{prop}
\begin{proof}
We first consider the generic flat functor: $y:\Box\to\psh{\Box}$.\\
Let $\I(n):=DL(n)$ be the generic free distributive lattice in $\psh{\Theta_{DL}}=\psh{\Box}$.
We have: $S_\I(n)=\I^n=\hom_{DL}(n,\I)$. The latter is the hom-set in
the Kleisli category. Since products are geometric,
$\I^n(m)=DL(m)^n=\hom_{DL}(n,m)$ for all $m$. Note that $y(n):=hom_\Box(-,n)=\hom_{DL}(n,-)$. So, $y=S_\I$!\\
As Johnstone~\cite{johnstone1979topological} observes (for simplicial
sets) this suffices. Let $f$ be flat functor, then
$f=F^\ast y$ for a geometric morphism $F$. Since the construction of
$S$ is geometric, as it only uses natural numbers and free
constructions, we have $f=F^\ast S_\I=S_{F^\ast\I}$.
\end{proof}

The previous proposition allows us to conclude a special case of Theorem~\ref{JW}.
\begin{thm}\label{class}
The topos $\psh{\Box}$ of cubical sets classifies flat distributive lattices.
\end{thm}

\subsection{Geometric realization}
In Theorem~\ref{geom} we will construct a cubical geometric realization.

\begin{prop}\label{linDL}
  Every linear order $D$ defines a flat distributive lattice. Hence, we have a geometric morphism $\psh{\Delta}\to\psh\Box$.
\end{prop}
\begin{proof}
  Obviously $D$ defines a distributive lattice. To check freeness,
  take $d\in D^m$ and $\alpha,\beta:n\to DL m$ as above.  Now choose $d'\in D^k$
  which orders the elements of $d$ and removes duplicates. Define
  $\gamma_1:D^k\to D^m$ such that $\gamma_1 d'=d$. This is possible
  because we have a (decidable) linear order. As $\alpha d=\beta d\in D^n$
  each coordinate function must be equal. Moreover, since $D$ is a
  linear order, the lattice operations reduce to one of their
  arguments and so the functions $\alpha,\beta$ are equal to a list of projection functions on $m$. 
%The only way projections of $d$ can be equal is when $d$ has duplicates.

  Consider the map $\gamma_2:k\to DL k$ defined by $\gamma_2(x_1,\ldots,x_k)=(x_1,x_1\vee
  x_2,\ldots, \bigvee_{i=1}^k x_i)$.  Then $\gamma_2(d')=d'$, as $d'$ is already
  ordered.  Let $\gamma=\gamma_1\gamma_2$. Each finite linear order is
  isomorphic to an interval in a free distributive lattice of the same
  size. In particular, $(x_1,x_1\vee x_2,\ldots, \bigvee_{i=1}^k x_i)$
  is isomorphic as a linear order to $d'$. Since $\alpha d= \alpha \gamma d'= \beta \gamma d'$,
  we have $\alpha\gamma=\beta\gamma$.  It follows that $D$ is free.
\end{proof}

Let $\mathcal{E}$ be Johnstone's topological topos.
\begin{thm}[Cubical geometric realization]\label{geom} There is a geometric morphism
  $r:\mathcal{E}\to\hat\Delta\to \hat\Box$ defined using $I:=[0,1]$ in
  $\mathcal{E}$. Moreover, $r^*$ is
  the following weighted colimit  
\[r^*K=K\otimes_\Box I^\bullet=\int^{n\in\Box}K(n)\times I^n.\]
\end{thm}
% http://ncatlab.org/nlab/show/cubical+set#GeometricRealization
\begin{proof}
  By analogy to the treatment by Johntone~\cite[Thm 8.1]{johnstone1979topological}, we define the
  cubical realization by using that $I$ in Top is a flat
  distributive lattice,
  using proposition~\ref{linDL}. By construction,
  $r^*K=S(I)\otimes_\Box K$, where we use
  Proposition~\ref{prop:flatfree} that $S(I)$ is a cocubical
  set; see~\cite[VI.5]{maclanemoerdijk}. Johnstone proves the preservation of
  colimits from $Top$ to $\mathcal{E}$ and we know that
  $S_n(I)=I^n$. This gives the formula above, where the right hand
  side is the usual coend formula ($S \otimes_C T = \int^c S(c) \otimes D(c)$) of the tensor product.
% https://ncatlab.org/nlab/show/end
% It is almost the
%  definition of the cubical realization in the nlab with the
%  difference that there one does not use diagonals.

 For every $X$ in $\mathcal{E}$, $r_*X(n)=\hom_\mathcal{E}(I^n,X)$. As the sequential
 spaces form a (reflective) subcategory of $\mathcal{E}$~\cite[Lem.\ 2.1]{johnstone1979topological}, this reduces to the
 cubical singular complex in the case of such spaces. This is
 reminiscent of the cubical realization studied by Jardine~\cite{jardine2002cubical}.
\end{proof}

Concretely, by the computation above, this geometric morphism $r$ is
represented by the flat functor $S_D(n)=D^n=y_{[1]}^n$. From this we
can compute the inverse image of the geometric morphism~\cite[B3.2.7]{elephant}.

Johnstone uses classical logic to prove that the unit interval is a
linear order. Without classical logic we can still define the
geometric realization of simplicial sets and of cubical sets. This
realization lands in the category $\mathcal{F}$ of so-called
sequential spaces, a reflective subcategory of $\mathcal{E}$.
\begin{prop}
The simplicial geometric realization $G$ is left exact iff $[0,1]$ is a
linear order. 
The cubical geometric realization is is left exact iff $[0,1]$ is a flat distributive lattice.
\end{prop}
\begin{proof}
If the geometric realization to spaces is left exact (i.e.\ preserves
limits), then so is the realization in $\mathcal{E}$, because $\mathcal{F}$ is a
reflective subcategory of $\mathcal{E}$ and hence the limits coincide.
Being left exact this functor is part of a geometric morphism
$\mathcal{E}\to \hat\Delta$. By the classifying property of
$\hat\Delta$, this means that $[0,1]$ is a linear order.

Conversely, if $[0,1]$ is a linear order, then
$r^*:\psh\Delta\to\mathcal{E}$ is left exact. Since it lands in $\mathcal{F}$
and the limits coincide, the geometric realization is left exact too.

A similar argument works for the cubical sets.
\end{proof}

By our construction, the cubical realization factors via the simplicial
realization. It is interesting to note that Jardine has geometrical
realizations going both into $\psh\Delta$ and into $Top$. His formulas for
the simplicial realization~\cite[p10]{jardine2002cubical} of a cubical
set are analogous to the one in Theorem~\ref{geom}.
% http://ncatlab.org/nlab/show/cubical+set#GeometricRealization

The cubical realization of a topological space in fact has
compositions~\cite{Cubical}, so it is constructively fibrant.

\subsubsection{Alternative combinatorial structures}
\paragraph{Simplicial sets}
Like the cube category, the simplex category can also be presented by
a duality. 
The opposite of the simplex category (of ordinals with face
and degeneracy maps) is the category $\nabla$ of intervals, finite
linear orders with distinct $\top$ and $\bot$ and monotone functions
preserving these. Here 2 is an dualizing object~\cite{Wraith}.
% http://ncatlab.org/nlab/show/simplex+category#DualityWithIntervals
The Yoneda embedding of this element gives the universal order in
sSets. It is called $V$ in~\cite[p458]{maclanemoerdijk}.
To connect this with the cubical sets, observe that
$\hom_{\Box}(-,\two)=\hom_{DL}(DL(1),-)\cong DL(-)$. For the
isomorphism in the previous sentence observe that such maps are determined by what happens to
the generator.

\paragraph{Bi-pointed sets} 
Awodey\footnote{\url{https://ncatlab.org/homotopytypetheory/files/AwodeyDMVrev.pdf}} uses cubes with diagonals (but without connections, or
reversions). This is Grothendieck's simplest test category. This cube category, $H$, is the free finite product category on
an interval. Awodey observes that $H=\Theta\op_2$, the Lawvere theory of bi-pointed sets.
As above, we also obtain a cubical realization for the topos $\psh H$ which classifies strictly bi-pointed sets. Moreover, this can
even be seen as a geometric morphism to the Giraud topos~\cite{johnstone1979topological}, roughly
sheaves over Top, as the
topological interval \emph{is} strictly bipointed, but not
constructively a strict
linear order.

We would like to relate
Coquand's and Awodey's cubical sets. There is a free distributive lattice on a bipointed set. As both theories are Cartesian,
the free distributive lattice over a bipointed set is geometric. We obtain a geometric
morphism from $\psh{H}\to\psh\Box$. This is the geometric morphism
obtained from the functor $H\op=\Theta\op_2\to \Theta\op_{DL}=\Box\op$ which is faithful, but not
full. Hence the geometric morphism is not an embedding.

% The geometric morphism
% $\psh{H}\to\psh\Box$ is also given by a flat functor
% $\Box\to\psh{H}$. 
% Is it $n,m\mapsto \hom_2(U DL(n),m)$, where $U$ forgets the lattice structure?
We obtain a geometric morphism in the opposite direction, $\psh\Box\to \psh{H}$, by observing that
every distributive lattice is bipointed. However, this is not the inverse
of the former map.

%Finally, in every free distributive lattice, the join-and-meet irreducible elements form
%a bipointed set. Hence, we obtain a geometric morphism
%$\psh\Box\to\psh{H}$. This construction is inverse to the free
%5distributive lattice,
%considered in the first map, but the construction is not geometric, so it is not (naturally) a Morita equivalence.%Olivia l-groups

\paragraph{De Morgan algebras and
  involutions}\label{sec:demorg-algebr-invol}
Since the cubical model uses De Morgan algebras, it is tempting to consider
the geometric realization for these cubical sets. 
Johnstone's argument shows that $[0,1]$ is also an \emph{involutive} linear
order with endpoints. Involutive means that we have extra rules: $a^{**}=a$, $a^*<b^*$ iff
$b<a$, and thus $a<b^*$ iff $b<a^*$. An involutive linear order is a
flat De Morgan algebra.

However, it seems that the classifying topos for involutive linear
orders has not been used in homotopy theory, so we will not pursue
this line further.

We would like to compare the toposes $\psh{\Theta_{DL}}$ and
$\psh{\Theta_{DM}}$. However, the obvious maps are not embeddings:
The functor $\Theta_{DL}\to\Theta_{DM}$ defined by $DL(n)\to DM(n)$ is
not full: the map $DM(1)\to DM(1)$ generated by $x\mapsto x^*$ is not
in the range. Likewise, the functor $\Theta_{DM}\to\Theta_{DL}$ defined by $DM(n)\to DL(2n)$ is
not full as we have $f(a^*)=f(a)^*$, so the generators are related by
the maps.

\section{Moore paths and identity types}
In this subsection, we focus on De Morgan algebras, instead of distributive
lattices. It makes for a slightly smoother presentation, but these
reversions are not strictly needed~\cite{Docherty}.
Coquand's presentation of the cubical model does not build on a
general categorical framework for constructing models of type theory.
Meanwhile a more abstract description has been
given~\cite{GCTT,ortonpitts} motivated by~\cite{Cubical:internal} and the
present work which was announced in~\cite{Spitters:TYPES}.

Here we pursue Docherty's model of identity types on cubical
sets with connections~\cite{Docherty}. This uses the general theory of path object
categories~\cite{vdBG}. We present a slightly different construction
from~\cite{Docherty} using similar tools, but
simplified by the use of internal reasoning, starting from the
observation that the generic De Morgan algebra $\I$ represents the interval
in cubical type theory.
To obtain a model of identity types on a category $C$ it suffices to
provide an involutive `Moore path' category object on $C$ with certain
properties. Now, category objects on cubical sets are 
categories in that topos. 
The Moore path category $MX$ consists of
lists of composable paths $\I\to X$ with the zero-length paths $e_x$ as
left and right identity. This is an instance of the general \emph{path
category} on a directed graph. In this case, the graph with edges given by
elements of $X^\I$. To obtain a \emph{nice} path object category ($M1\cong1$),
we quotient by the relation which removes lists of constant paths
$(\lambda \_.x)$ from
the list; cf~\cite[Def 3.9]{Docherty}. Since $1^\I\cong1$ the only Moore
path in 1 is in fact the zero-length path. Hence, $M1\cong1$.

The reversion $\neg$ on $\I$ allows us to reverse
paths of length 1. This reversion extends to paths of any length. We obtain an
involutive category: Moore paths provide strictly associative
composition, but non-strict inverses.

A path contraction is a map $\operatorname{con}:MX\to MMX$ which maps a path $p$ to a
path from $p$ to the constant path on $tp$ ($t$ for target).
Like Docherty~\cite[Def 3.2.3]{Docherty}, we use connections to first define the map $con_1$ from $X^\I$ to $X^{\I\times
  \I}$ by $\lambda p. \lambda i j. p(i\vee j)$ and then extended this
to a contraction. For a Moore path of length three, this looks like:
\begin{center}\small{
\includegraphics{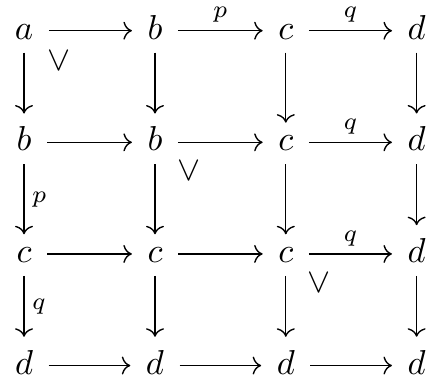}}
% \begin{tikzcd}
% a \arrow[r]\arrow[d] \arrow[dr, phantom,"\vee", very near start]& b \arrow[d] \arrow[r,"p"]&c\arrow[d]\arrow[r,"q"]&d\arrow[d]\\
% b \arrow[r]\arrow[d,"p"]& b\arrow[d] \arrow[r] \arrow[dr, phantom,"\vee", very near start]&c\arrow[d]\arrow[r,"q"]&d\arrow[d]\\
% c \arrow[r]\arrow[d,"q"] & c \arrow[r]\arrow[d]&c\arrow[d]\arrow[r,"q"]\arrow[dr, phantom,"\vee", very near start]&d\arrow[d]\\
% d \arrow[r] & d \arrow[r]&d \arrow[r]&d\\
% \end{tikzcd}}
\end{center}%
All these constructions are algebraic and hence work functorially. We
obtain a nice path object category.

We have obtained a model of identity types~\cite{vdBG,Docherty} starting from
the interval $\I$ in the cubical model. The connection between this model and the one by Coquand is
not entirely clear. We will explore what can be done when we restrict
to the fibrant types, the types with composition structure.

\subsection{Id, Path and Moore}
In this section we study the relation between the three constructions
of identity types. The Path type $-^\I$ is the main identity type in
cubical type theory. However, it does not satisfy the judgemental
computation rule for the $J$-elimination. From $Path$ one can define a
new type $Id$ which does satisfy the judgemental rule~\cite[9.1]{Cubical}.
   \begin{mathpar}
    \inferrule{\Gamma,\phi\vdash a:A\qquad \Gamma\vdash p:\Path_A a b [\phi\mapsto
      \langle{\um}\rangle a]
    }{\Gamma\vdash (p,[\phi\mapsto a]):\id_A a b}
  \end{mathpar}
\paragraph{Id and Path}
We start with some short observations about Path and Id. The section $\lambda p.(p,0): \Path \to \id$ has $\pi_1$ as a
(judgemental) retraction. Moreover, it has been formally checked that they give rise a
Path-quasi inverse%
\footnote{\texttt{idToEqToId}, \texttt{eqToIdToEq} in \url{https://github.com/mortberg/cubicaltt/tree/defeq}}. This can be weakened to an Id-quasi-inverse.

The proofs that $\Pi,\Sigma$ respect equivalence do not require univalence.
It follows that all the predicates which are defined from an equality are
equivalent with respect to either notion of equality. In particular, this holds for
being contractible, being a proposition and being an equivalence. 

\paragraph{Moore paths}
We will show that for each fibrant type $A$, there is a new fibrant
type $M_0A$ of Moore paths. Cubical type theory supports inductive types and since it has a identity type, it is
expected to support inductive families~\cite{dybjer1994inductive} too.
However, this has not been formally checked yet. Instead, we can
define the composition explicitly for one particular inductive family.
This is a slight variation on the transitive reflexive closure of the
Path relation.
We make the following definition in cubical \emph{sets} and will show that it
has a composition operation.
\begin{eqnarray*} 
\mathrm{Inductive}\ M_0A (x:A) &:& A \to \operatorname{Type} :=\\
  e&:&\Pi_{a:A}M_0A a a\\
  c&:&\Pi_{p:PA}\Pi_{l:M_0 A (p1) a}M_0A (p0) a
\end{eqnarray*}
We will write $M_0A$ for $\Sigma_{a b}M_0A\ a\ b$. As we've seen above in
the definition of contraction for Moore paths, paths between Moore
paths will preserve the length, but may vary the points.

First we recall the definition of composition on lists. It is defined
recursively by:
\begin{eqnarray*}
 \comp^i (\List A) [\phi \mapsto \nil] \nil &= & \nil \\
 \comp^i (\List A) [\phi \mapsto a::l ] a'::l' &=  &
 \comp^i A [\phi \mapsto a ] a' :: \comp^i (\List A) [\phi \mapsto l ] l'
\end{eqnarray*}

Similarly,
\begin{eqnarray*}
 \comp^i M_0A [\phi \mapsto e_a] e_{a'} &= & e_{\comp^i A [\phi \mapsto a] a'} \\
 \comp^i M_0A [\phi \mapsto p::l ] p'::l' &=  &
 \comp^i PA [\phi \mapsto p ] p' :: \comp^i M_0A [\phi \mapsto l ] l'
\end{eqnarray*}
Where we have written $p::l$ for the concatenation $c p l$.

To obtain a nice path category, we would like to quotient Moore paths upto constant paths; as in~\cite[Dfn
3.9]{Docherty}. Unfortunately, constancy of paths is not decidable and
neither is it respected by comp. So, this approach seems stuck here. As Docherty~\cite[Ch7]{Docherty} points out, it may be
possible to relax the requirement that $M1\cong1$. So, it seems worth
recording the facts above, even if they are not yet conclusive.
\emph{If} such Moore paths can be defined on the fibrant types, they will have the judgemental computation rules for
$J$, and being equivalent type Path, will also satisfy univalence,
since~\cite[Cor.11]{Cubical} works for any other notion of identity which is
reflexive and satisfies the elimination rule for identity.

\section{Pre-model-structure}
The cubical model may be constructed in the internal language of the topos
of cubical sets~\cite{Cubical:internal,GCTT, ortonpitts}. 
As recalled in the previous section, one can introduce a
type Id which has a judgemental computation rule for the $J$
eliminator. A factorization system can be defined in the cubical type theory~\cite[9.1]{Cubical}.
The Gambino-Garner factorization system~\cite{gambino2008identity}
gives another factorization system. We will show that the two factorization systems together
form a pre-model-structure. Here we follow Lumsdaine~\cite{Lumsdaine:premodel}, with the twist
that he uses Id-types, rather than Path-types for the mapping
cylinder. 

To be precise, Gambino-Garner require identity
\emph{contexts}. Since the cubical model is only a
category with families, we need to consider the associated contextual category. I.e.\ we need to restrict to
the \emph{fibrant} cubical sets.

\begin{defn} A \emph{pre-model-structure} on a category $C$ consists of three classes
$(\mathcal{C}, \mathcal{F}, \mathcal{W})$ of maps of $C$,
such that $\mathcal{W}$ satisfies 3-for-2, and $(\mathcal{C},
\mathcal{F}\cap \mathcal{W})$ and $(\mathcal{C}\cap\mathcal{W}, \mathcal{F})$
are weak factorisation systems on $\mathcal{C}$.
\end{defn}

\begin{thm}
There is a pre-model-structure on \emph{fibrant} cubical sets. The factorizations are defined in the
cubical type theory, so they are uniform and functorial.
\end{thm}
\begin{proof}
We define $\mathcal{F}_0$ to be the set of display maps
$\pi_1:\Sigma AB\to A$. Let $\mathcal{W}$ be the set of equivalences.
Let $\mathcal{TF}_0$ be the display maps which are also equivalences,
i.e.\ every fiber is \emph{Path}-contractible. The cofibrations $\mathcal C$ are the
maps with the lift lifting property (LLP) with respect to
$\mathcal{TF}_0$. 
This LLP is wrt diagrams which commute upto judgemental equality,
equality in the topos, not just upto a Path.
We define the trivial cofibrations $\mathcal{TC}=\mathcal{W}\cap\mathcal{C}$.

Gambino-Garner factors a function $f:A\to B$
through its graph
$\Sigma_{(y;x)}\id_B(f(x),y)$. The maps $tc_f(a):=(f(x);x;1_{fx})$ and
$\pi_1$ give the factorization of $f$. The map $tc_f$ is a section
with retraction $\pi_1\pi_2$. This shows that $tc_f\in \mathcal{TC}$.

The other factorization is given by the type $C_f:=\Sigma_bT_f(b)$,
where $T_f(b)$ consists of partial sections $[\phi\mapsto a]$ such that
$f(a)=b$ on $\phi$. The type $C_f$ is reminiscent of the mapping cylinder, the homotopy pushout of $1$ and
$f$, a higher inductive type (HIT) equivalent to $B$; see~\cite{hottbook}. While there is no
general theory of HITs in cubical yet, we can explicitly define the constructors of this HIT:\\
$\mathsf{inbase}\ b:=(b,[0\mapsto a_0]):C_f(b)$,\\
$\mathsf{intop}\ a:=(f(a),[1\mapsto a]):C_f(f a)$ and\\
$\mathsf{incyl}\ a:=(\langle\um\rangle f(a),a_0):\Path_{C_f}\mathsf{intop} a,
\mathsf{inbase} (f a)$.

Coquand shows that we have a factorization $(\mathcal{C},
\mathcal{TF}_0)$.

As Lumsdaine shows $\mathcal{W}$ has 3-for-2 and is closed under
retracts. To be precise, these are retracts between fibrant objects,
hence judgement equalities in the type theory. 

We can close the classes $\mathcal{F}_0$ and $\mathcal{TF}_0$ by
taking the double orthogonal, where being orthogonal is defined in the cubical type theory.

Lumsdaine shows that $\mathcal{TF}=\mathcal{TF}\cap\mathcal{W}$ and
$\mathcal{TC}=\mathcal{C}\cap\mathcal{W}$ by formal manipulations
about retracts and orthogonality. The same arguments go through here.  
\end{proof}
\begin{acknowledgement}
I would like to thank Steve Awodey and Thierry Coquand for discussions
on the topic of this paper.
Part of this work was done while the author was working at
Gothenburg University and, later, at Carnegie-Mellon University.\\
This research was partially supported by the Guarded homotopy type theory project, funded by the Villum Foundation, project number 12386.
\end{acknowledgement}
\bibliographystyle{alpha}
\bibliography{class}
\end{document}